\documentclass[11pt,english,letterpaper,article,superscriptaddress,nofootinbib]{revtex4}
\usepackage[T1]{fontenc}
\usepackage[latin9]{inputenc}
\usepackage{prettyref}
\usepackage{float}
\usepackage{framed}
\usepackage{amsthm}
\usepackage{amsmath}
\usepackage{graphicx}

\makeatletter
\@ifundefined{textcolor}{}
{%
 \definecolor{BLACK}{gray}{0}
 \definecolor{WHITE}{gray}{1}
 \definecolor{RED}{rgb}{1,0,0}
 \definecolor{GREEN}{rgb}{0,1,0}
 \definecolor{BLUE}{rgb}{0,0,1}
 \definecolor{CYAN}{cmyk}{1,0,0,0}
 \definecolor{MAGENTA}{cmyk}{0,1,0,0}
 \definecolor{YELLOW}{cmyk}{0,0,1,0}
 }
  \theoremstyle{plain}
  \newtheorem{thm}{Theorem}
  \theoremstyle{plain}
  \newtheorem{lem}{Lemma}
 \ifx\proof\undefined\
   \newenvironment{proof}[1][\proofname]{\par
     \normalfont\topsep6\p@\@plus6\p@\relax
     \trivlist
     \itemindent\parindent
     \item[\hskip\labelsep
           \scshape
       #1]\ignorespaces
   }{%
     \endtrivlist\@endpefalse
   }
   \providecommand{\proofname}{Proof}
 \fi

\usepackage{amsthm}
\usepackage{prettyref}
\newrefformat{thm}{theorem~\ref{#1}}
\newrefformat{lem}{~\eqref{#1}}
\newrefformat{eq}{~\eqref{#1}}

\newrefformat{lem}{~\eqref{#1}}

\makeatother

\usepackage{babel}

\makeatother

\usepackage{babel}
\newtheorem*{adiabaticthm}{Adiabatic Theorem (from \cite{2007JMP....48j2111J})}
\newtheorem*{scramblethm}{Scrambled Theorem (modified from \cite{2005quant.ph.12159F})}

\makeatother

\usepackage{babel}

\begin{document}
\title{Unstructured Randomness, Small Gaps and Localization}
\author{Edward Farhi}
\affiliation{Center for Theoretical Physics, Massachusetts Institute of Technology, Cambridge, MA 02139}
\author{Jeffrey Goldstone}
\affiliation{Center for Theoretical Physics, Massachusetts Institute of Technology, Cambridge, MA 02139}
\author{David Gosset}
\affiliation{Center for Theoretical Physics, Massachusetts Institute of Technology, Cambridge, MA 02139}
\author{Sam Gutmann}
\noaffiliation
\author{Peter Shor}
\affiliation{Center for Theoretical Physics, Massachusetts Institute of Technology, Cambridge, MA 02139}
\affiliation{Department of Mathematics, Massachusetts Institute of Technology, Cambridge, MA 02139}
\begin{abstract} We study the Hamiltonian associated with the quantum adiabatic algorithm with a random cost function. Because the cost function lacks structure we can prove results about the ground state. We find the ground state energy as the number of bits goes to infinity, show that the minimum gap goes to zero exponentially quickly, and we see a localization transition. We prove that there are no levels approaching the ground state near the end of the evolution. We do not know which features of this model are shared by a quantum adiabatic algorithm applied to random instances of satisfiability since despite being random they do have bit structure.
\end{abstract}
\maketitle

\section{Introduction, Discussion, and Conclusions}

Recently there has been interest in the relevance of Anderson localization
to the quantum adiabatic algorithm\cite{2009arXiv0908.2782A,2009arXiv0912.0746A,2010arXiv1005.3011K}.
In this paper we study properties of the Hamiltonian associated with
the adiabatic algorithm with a cost function that has random (essentially)
uncorrelated values. The cost function we look at does not have the
structure present in cost functions produced by random instances of
satisfiability. This lack of structure makes our example analyzable
and the localization transition and corresponding small gap will be
evident. The model is a spin Hamiltonian on $n$ spins that takes
the form \begin{equation}
H(s)=\left(1-s\right)\sum_{i=1}^{n}\left(\frac{1-\sigma_{x}^{i}}{2}\right)+s\sum_{z=0}^{2^{n}-1}E(z)|z\rangle\langle z|.\label{eq:h}\end{equation}
The {}``on site energies'' $E(z)$ are random variables obtained
by scrambling the Hamming weight cost function. Viewing $z$ in\prettyref{eq:h}
as an n bit string, the Hamming weight $W(z)$ is the number of ones
in the string. Let $\pi$ be a random permutation of the $2^{n}$
integers between $0$ and $2^{n}-1$. By ``random permutation'' we
mean that all $2^{n}!$ permutations are equally likely. Note we are
permuting the $2^{n}$ strings, not the $n$ bits. Then \begin{equation}
E(z)=W(\pi^{-1}(z))\label{eq:E(z)}\end{equation}
 and\begin{equation}
H_{\pi}(s)=\left(1-s\right)\sum_{i=1}^{n}\left(\frac{1-\sigma_{x}^{i}}{2}\right)+s\sum_{z=0}^{2^{n}-1}W(\pi^{-1}(z))|z\rangle\langle z|.\label{eq:hpi}\end{equation}

In figure \ref{fig:picture} we plot the lowest 25 energy levels of
this Hamiltonian, as a function of $s$, with $n=18$ and with a particular
random choice of the permutation $\pi$. Different random permutations
produce very similar pictures. The ground state energy is well approximated
by two straight lines\cite{2005quant.ph.12159F,2010PThPS.184..290J,2006JSMTE..11..012O}.
We will prove that asymptotically (in $n)$ they are in fact straight
lines. At $s=\frac{1}{2}$ the gap between the ground state and the
first excited state is small and we will show that asymptotically
the minimum gap is exponentially small in $n$. 

This Hamiltonian is similar to the {}``random energy model'' of
the form\prettyref{eq:h} with on site energies chosen i.i.d, that
is, independently and identically distributed. Some of what we prove
is already known for the random energy model\cite{2006JSMTE..11..012O,2010arXiv1002.4409P,2010PThPS.184..290J}.
For example, in \cite{2010PThPS.184..290J}, Jörg et. al. use perturbation
theory to show the existence of a first order phase transition in
a random energy model. In contrast, we use variational methods which
prove upper and lower bounds on the ground state energy. To prove
that the gap is exponentially small in $n$, we use the information
theoretic result\cite{2005quant.ph.12159F} that no efficient quantum
algorithm exists for locating the minimum of a scrambled cost function.

It is apparent from the picture that the ground state changes dramatically
at $s=\frac{1}{2}$. In fact for $s<\frac{1}{2}$, the ground state
is very close to the ground state at $s=0$,\[
|x=0\rangle=\frac{1}{\sqrt{2^{n}}}\sum_{z=0}^{2^{n}-1}|z\rangle\]
 which in the $z$ basis is completely delocalized. For $s>\frac{1}{2}$
the ground state is close to the ground state at $s=1$ which is \[
|z=\pi(0)\rangle\]
 corresponding to one $z$ string, that is, a fully localized state. 

The small gap just seen at $s=\frac{1}{2}$ is associated with a {}``delocalized
to localized'' first order phase transition. Recently it has been
suggested that a {}``localized to localized'' ground state transition
can also lead to an exponentially small gap. These ground state transitions,
studied in \cite{2009arXiv0908.2782A,2009arXiv0912.0746A,2009arXiv0909.4766F,2010arXiv1005.3011K,2009PhRvA..80f2326A},
can be seen using low order perturbation theory, and occur for $s\rightarrow1$
as $n\rightarrow\infty$. There are two key features of a model that
exhibits these perturbative crosses. First, for any string a single
bit flip changes the cost function by $O(1)$. The other feature is
that there are very disparate bit strings with low cost. The scrambled
model that we study has the second feature but not the first and we
show that perturbative crosses are not present. We prove that for
$s\geq0.9$ the gap is greater than a positive constant independent
of $n$. 

\begin{figure}[t]
\begin{centering}
\includegraphics[scale=0.7]{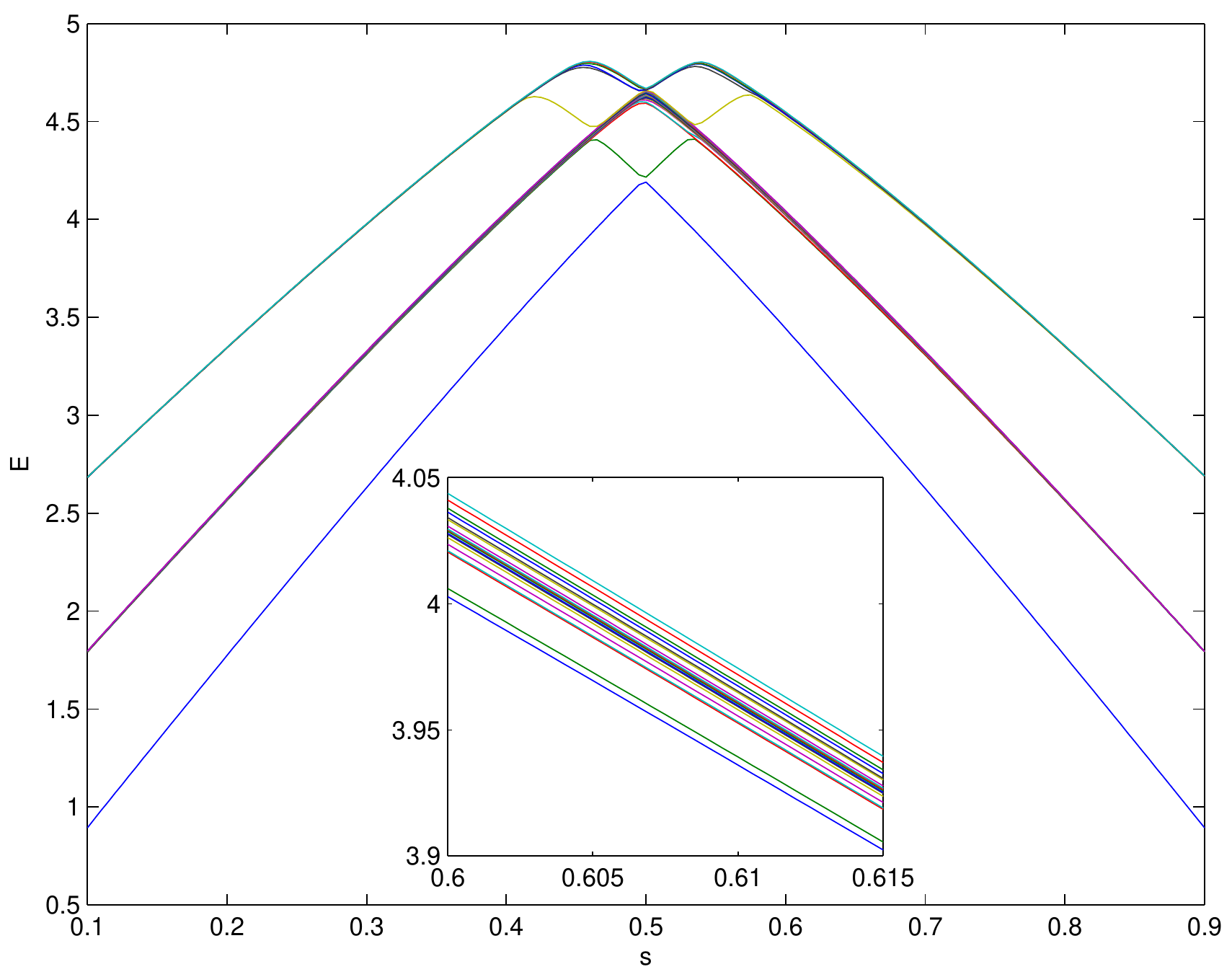}
\par\end{centering}

\caption{Lowest 25 energy levels for an instance with a random permutation
at $n=18$. The inset shows a magnified view of levels 2 through 19
near $s=0.6$.\label{fig:picture} }

\end{figure}

\begin{figure}[H]
\begin{centering}
\includegraphics[scale=0.7]{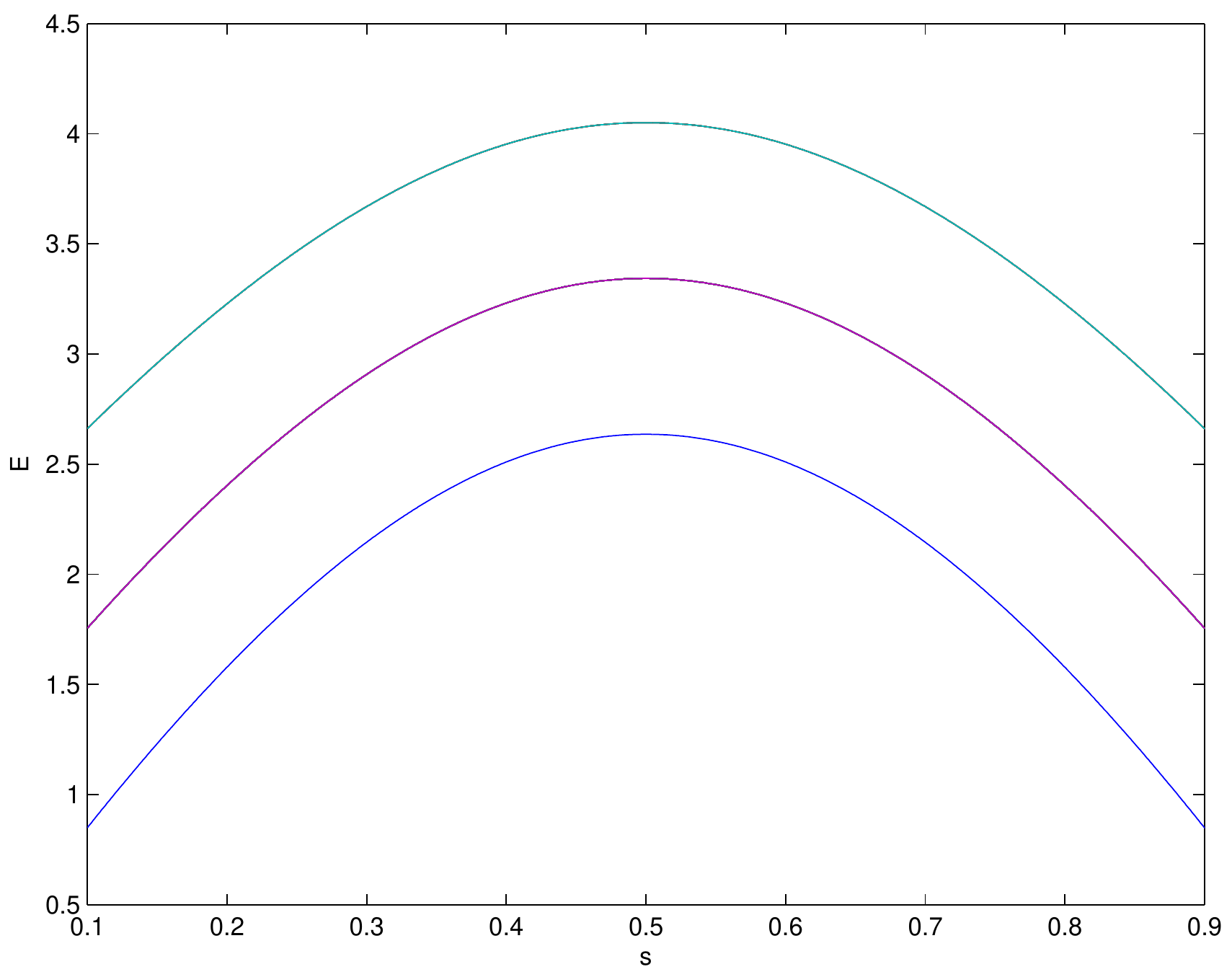}
\par\end{centering}

\caption{Lowest 25 levels for the instance with the identity permutation at
$n=18$.\label{fig:decoupled}}

\end{figure}

The regularity in figure \ref{fig:picture} and the simple description
of the ground state are consequences of the total lack of structure
in the scrambled cost function. Although it is not relevant to the
performance of the adiabatic algorithm, it is interesting to look
at the middle of the spectrum. In figure \ref{fig:Ten-energy-levels}
we show energy levels in the middle of the spectrum for a random choice
of $\pi$ at $n=14$. Here there are many avoided crosses and no apparent
regularity in contrast to the bottom of the spectrum.

Suppose the permutation $\pi$ used to generate figure \ref{fig:picture}
were replaced by the identity. Now the cost function, $E(z)=W(z)$,
is the unscrambled Hamming weight and the Hamiltonian\prettyref{eq:hpi}
is a sum of $n$ one qubit Hamiltonians making it easy to analyze.
The lowest $25$ levels for $n=18$ are shown in figure \ref{fig:decoupled}.
Note that the minimum gap is independent of $n$. Furthermore the
values and degeneracies of the on site energies in figure \ref{fig:decoupled}
are the same as those in figure \ref{fig:picture}. The radical differences
between these figures are caused by the bit structure present in the
unscrambled case but absent in the scrambled case.

\begin{figure}
\includegraphics[scale=0.75]{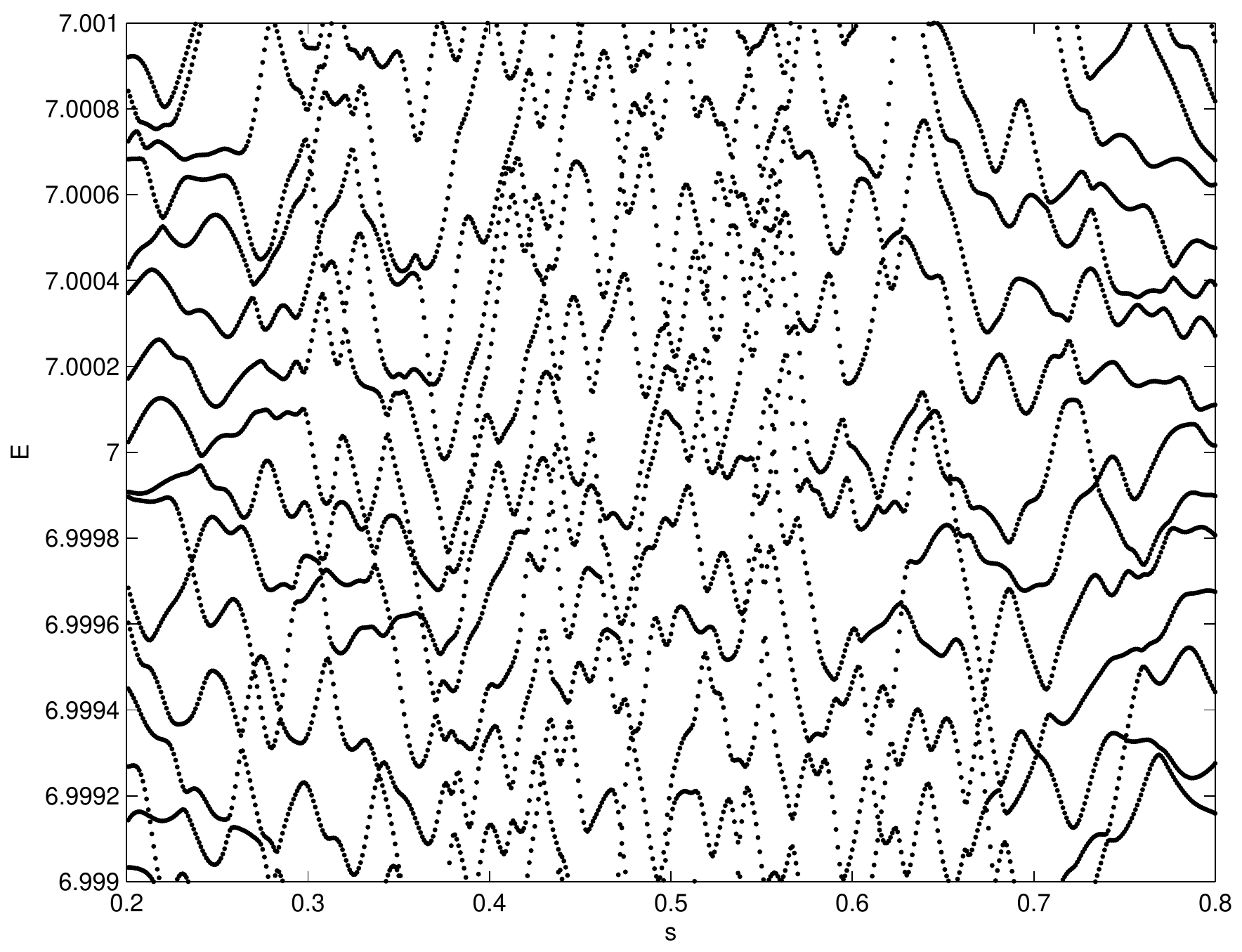}

\caption{Energy levels in the middle of the spectrum for a 14 spin instance
with a random permutation.\label{fig:Ten-energy-levels}}

\end{figure}
It is not yet known how the quantum adiabatic algorithm performs on
random instances of satisfiability, for example, instances of Exact
Cover generated by random choices of 3 bits \cite{2008PhRvL.101q0503Y,2010PhRvL.104b0502Y}.
Despite the randomness, these Hamiltonians still have bit structure.
We do not know if random instances of satisfiability have a {}``delocalized
to localized'' first order phase transition like the fully unstructured
example shown in figure \ref{fig:picture} or have no {}``delocalized
to localized'' phase transition like the bit structured example shown
in figure \ref{fig:decoupled}.

\section{Results}

For a given $n$ and $\pi$ we have $H_{\pi}(s)$ given by\prettyref{eq:hpi}.
Denote the ground state energy by $E_{\pi}(s)$. Also let\begin{equation}
e(s)=\begin{cases}
\frac{s}{2} & \text{ for }0\leq s\leq\frac{1}{2}\\
\frac{\left(1-s\right)}{2} & \text{ for }\frac{1}{2}\leq s\leq1.\end{cases}\label{eq:e(s)}\end{equation}

We have 3 results which suffice for our needs but are not best possible.
First $\frac{E_{\pi}(s)}{n}$ is well approximated by $e(s)$ for
$n$ large. %
\begin{framed}%
\begin{thm}
\label{thm:mainthm}Let $\pi$ be a random permutation. Then \[
\frac{E_{\pi}(s)}{n}\leq e(s)\text{ for all }0\leq s\leq1\]
 and \[
\text{Pr}\left[\frac{E_{\pi}(s)}{n}\geq e(s)\left[1-5n^{-\frac{1}{4}}\right]\text{ for all }0\leq s\leq1\right]\rightarrow1\text{ as }n\rightarrow\infty.\]

\end{thm}
\end{framed}

By $Pr\left[event\right]$ we mean the probability that the \textit{event
}occurs where the distribution is over $\pi$ chosen with equal probability
from the $2^{n}!$ permutations.

Our next result is that the minimum gap is exponentially small in
$n$. For each $\pi$ let \[
g_{\pi}=\min_{0\leq s\leq1}\gamma_{\pi}(s).\]
 where $\gamma_{\pi}(s)$ is the eigenvalue gap between the two lowest
energy levels of $H_{\pi}(s)$. %
\begin{framed}%
\begin{thm}
\label{thm:gapthm}Let $\pi$ be a random permutation. Then \[
\text{Pr}\left[g_{\pi}\leq n^{2}2^{-\frac{n}{6}}\right]\rightarrow1\text{ as }n\rightarrow\infty.\]

\end{thm}
\end{framed}

Our last result is that typically $\gamma_{\pi}(s)$ is larger than
some $n$ independent constant for $s$ near $1$.

\begin{framed}%
\begin{thm}
\label{thm:nopert_cross}

Let $\pi$ be a random permutation. Then \[
\text{Pr}\left[\gamma_{\pi}(s)>0.8\text{ for }.9\leq s\leq1\right]\rightarrow1\text{ as \ensuremath{n\rightarrow\infty.}}\]

\end{thm}
\end{framed}

\subsection{A Lower Bound For the Ground State Energy of Stoquastic Hamiltonians}

For any normalized state $|\psi\rangle,$ and any Hamiltonian $H$,
the quantity $\langle\psi|H|\psi\rangle$ is an upper bound on the
ground state energy. It is less well known that, for some special
Hamiltonians, one can also obtain a variational lower bound on the
ground state energy of $H$. Theorem \ref{thm:collatz-wielandt} (given
below) is an example of such a bound. It is an elementary application
of the Collatz Wielandt min-max formula \cite{meyer01}, but we prove
it here for completeness. We restrict to Hamiltonian matrices $H$
where all off diagonal matrix elements are real and nonpositive. These
are called stoquastic Hamiltonians (note that stoquasticity is a basis
dependent notion).
\begin{thm}
\label{thm:collatz-wielandt}Let $H$ be a Hermitian operator, stoquastic
in the $|z\rangle$ basis, that is $\langle z|H|z^{\prime}\rangle\leq0$
for $z\neq z^{\prime}.$ Let $E_{g}$ be its lowest eigenvalue. Then
\[
E_{g}\geq\min_{z}\frac{\langle z|H|\phi\rangle}{\langle z|\phi\rangle}\]
 for any state $|\phi\rangle$ such that $\langle z|\phi\rangle>0$
for all $z$. 
\end{thm}
The following Lemma will be used in our proof of Theorem \ref{thm:collatz-wielandt}.
\begin{lem}
\label{lem:PF}Let $H$ be a Hermitian operator, stoquastic in the
$|z\rangle$ basis, that is $\langle z|H|z^{\prime}\rangle\leq0$
for $z\neq z^{\prime}.$ If $H|\psi\rangle=E|\psi\rangle$ and $\langle z|\psi\rangle>0$
for all $z$, then $E=E_{g}$, the ground state energy of $H.$ \end{lem}
\begin{proof}
Suppose $E>E_{g}$. Let $|\psi_{g}\rangle$ be any (normalized) state
in the ground state subspace of $H$, so that $H|\psi_{g}\rangle=E_{g}|\psi_{g}\rangle$.
Divide the values of $z$ into $3$ sets:\begin{eqnarray*}
S_{+} & = & \left\{ z:\:\langle z|\psi_{g}\rangle>0\right\} \\
S_{-} & = & \left\{ z:\:\langle z|\psi_{g}\rangle<0\right\} \\
S_{0} & = & \left\{ z:\:\langle z|\psi_{g}\rangle=0\right\} .\end{eqnarray*}
Define $|\psi_{g}^{\prime}\rangle$ by \[
\langle z|\psi_{g}^{\prime}\rangle=\left|\langle z|\psi_{g}\rangle\right|.\]
Then $|\psi_{g}^{\prime}\rangle$ is normalized and by the variational
principle \[
\langle\psi_{g}^{\prime}|H|\psi_{g}^{\prime}\rangle-\langle\psi_{g}|H|\psi_{g}\rangle\geq0.\]
 But \begin{eqnarray*}
\langle\psi_{g}^{\prime}|H|\psi_{g}^{\prime}\rangle-\langle\psi_{g}|H|\psi_{g}\rangle & = & -4\sum_{z\in S_{+}}\sum_{z^{\prime}\in S_{-}}\langle z|\psi_{g}\rangle\langle z^{\prime}|\psi_{g}\rangle\langle z|H|z^{\prime}\rangle\\
 & \leq & 0\end{eqnarray*}
where in the last line we have used the stoquasticity of $H$. Therefore
\[
\langle\psi_{g}^{\prime}|H|\psi_{g}^{\prime}\rangle=\langle\psi_{g}|H|\psi_{g}\rangle=E_{g}\]
 and so $H|\psi_{g}^{\prime}\rangle=E_{g}|\psi_{g}^{\prime}\rangle$.
This implies $\langle\psi|\psi_{g}^{\prime}\rangle=0$, which is impossible
since $|\psi\rangle$ has only positive coefficients in the $z$ basis
and $|\psi_{g}^{\prime}\rangle$ has nonnegative coefficients. So
we have reached a contradiction and therefore $E=E_{g}$.
\end{proof}

\subsection*{Proof of Theorem \ref{thm:collatz-wielandt}}

Let $\hat{H}$ be \[
\hat{H}=H-\sum_{z}|z\rangle\langle z|\frac{\langle z|H|\phi\rangle}{\langle z|\phi\rangle}.\]
 Note that $|\phi\rangle$ is an eigenvector of $\hat{H}$ with eigenvalue
$0$. The ground state energy of $\hat{H}$ is $0$ by Lemma \ref{lem:PF}. 

For any normalized state $|\psi\rangle$ we have \begin{eqnarray*}
\langle\psi|H|\psi\rangle-\langle\psi|\hat{H}|\psi\rangle & = & \sum_{z}\left|\langle z|\psi\rangle\right|^{2}\left[\frac{\langle z|H|\phi\rangle}{\langle z|\phi\rangle}\right]\\
 & \geq & \min_{z}\left[\frac{\langle z|H|\phi\rangle}{\langle z|\phi\rangle}\right].\end{eqnarray*}
Choosing $|\psi\rangle$ to be a ground state of $H$ we have \[
E_{g}-\langle\psi|\hat{H}|\psi\rangle\geq\min_{z}\left[\frac{\langle z|H|\phi\rangle}{\langle z|\phi\rangle}\right].\]
But $\langle\psi|\hat{H}|\psi\rangle\geq0$ since $\hat{H}$ has ground
state energy $0$, so \begin{eqnarray*}
E_{g} & \geq & \min_{z}\frac{\langle z|H|\phi\rangle}{\langle z|\phi\rangle}.\end{eqnarray*}

\subsection{Proof of Theorem \ref{thm:mainthm}}

First for the upper bound. Recall that \[
|x=0\rangle=\frac{1}{\sqrt{2^{n}}}\sum_{z}|z\rangle\]
 is the ground state of $H_{\pi}(0)$. Using this state we get an
upper bound on the ground state energy \begin{equation}
E_{\pi}(s)\leq\langle x=0|H_{\pi}(s)|x=0\rangle=s\frac{n}{2}.\label{eq:upperbnd_E}\end{equation}
 Recall that $|\pi(0)\rangle$ is the ground state of $H_{\pi}(1).$
Using this state we get the upper bound\[
E_{\pi}(s)\leq\langle\pi(0)|H_{\pi}(s)|\pi(0)\rangle=\left(1-s\right)\frac{n}{2}\]
for $0\leq s\leq1.$ These two inequalities establish the upper bound
in Theorem \ref{thm:mainthm}.

Proving the lower bound is more involved. We first use Theorem \ref{thm:collatz-wielandt}
to show that with high probability the ground state energy $E_{\pi}(s^{\star})$
at $s^{\star}$ just above $\frac{1}{2}$ is close to the value $n\cdot e(s^{\star})$.
Then we use the concavity of $E_{\pi}(s)$ to obtain the lower bound
for all $s\in[0,1]$.

We use an ansatz $|A_{\pi}\rangle$ for the ground state of $H_{\pi}(s^{\star})$
of the form \begin{equation}
\langle z|A_{\pi}\rangle=\begin{cases}
1 & ,\text{ z}\notin S_{\pi}\\
\lambda n & ,\text{ z}\in S_{\pi}\end{cases}\label{eq:API}\end{equation}
 where $\lambda$ satisfies $\lambda n>1$ for $n$ large and $S_{\pi}$
is a set of low energy bit strings, \[
S_{\pi}=\left\{ z:\: W(\pi^{-1}(z))\leq nc\right\} \]
 with $c<\frac{1}{2}$. We can view $S_{\pi}$ as the image under
$\pi$ of the set of strings with Hamming weight less than $nc$:\begin{eqnarray*}
S_{\pi} & = & \pi(L)\\
L & = & \left\{ z:\: W(z)\leq nc\right\} .\end{eqnarray*}
 The following lemma bounds the probability that a group of bit strings
related by single bit flips are all in a random set.
\begin{lem}
\label{lem:bitlemma}Let M be a set of $n$-bit strings such that
\[
\left|M\right|\leq2^{n\gamma}\]
 for some $0<\gamma<1$. Let k be an integer. Let $\pi$ be a random
permutation of the set of all $n$-bit strings, and consider $\pi(M)$,
the image of the set $M$. Let $p$ be the probability that there
exists an $n$-bit string $z$ and a set of $k-1$ bits $\left\{ i_{1},i_{2},...,i_{k-1}\right\} $
such that $\left\{ z,z\oplus e_{i_{1}},...,z\oplus e_{i_{k-1}}\right\} \subseteq\pi(M).$
Similarly let $q$ be the probability that there exists an $n$-bit
string $y$ and a set of $k$ bits $\left\{ j_{1},j_{2},...,j_{k}\right\} $
such that $\left\{ y\oplus e_{j_{1}},...,y\oplus e_{j_{k}}\right\} \subseteq\pi(M).$
Then \[
p\leq n^{k}2^{n\left[1-k\left(1-\gamma\right)\right]}\]
 and \[
q\leq n^{k}2^{n\left[1-k\left(1-\gamma\right)\right]}\:.\]
 \end{lem}
\begin{proof}
Fix a particular $z$ and a set of $k-1$ bits $\left\{ i_{1},i_{2},...,i_{k-1}\right\} $.
Then \begin{eqnarray*}
\text{Pr}\left[\left\{ z,z\oplus e_{i_{1}},...,z\oplus e_{i_{k-1}}\right\} \subseteq\pi(M)\right] & = & \frac{\left|M\right|\cdot\left(\left|M\right|-1\right)...\left(\left|M\right|-k+1\right)}{2^{n}\cdot\left(2^{n}-1\right)...\left(2^{n}-k+1\right)}\\
 & \leq & \left(\frac{\left|M\right|}{2^{n}}\right)^{k}.\end{eqnarray*}
 There are $2^{n}$ choices for $z$ and $\bigg(\begin{array}{c}
n\\
k-1\end{array}\bigg)$ choices for the set of $k-1$ bits so by a union bound \begin{eqnarray*}
p & \leq & 2^{n}\bigg(\begin{array}{c}
n\\
k-1\end{array}\bigg)\left(\frac{\left|M\right|}{2^{n}}\right)^{k}\\
 & \leq & n^{k}2^{n\left[1-k\left(1-\gamma\right)\right]}\:.\end{eqnarray*}
 The proof for $q$ is very similar.\end{proof}
\begin{lem}
\label{lem:Lemma}Let $s^{\star}=\frac{1}{2}+n^{-\frac{1}{4}}$. Let
$\pi$ be a random permutation. Then \[
\text{Pr}\left[E_{\pi}(s^{\star})\geq\left(\frac{1-s^{\star}}{2}\right)\left(n-n^{\frac{3}{4}}\right)\right]\rightarrow1\text{ as }n\rightarrow\infty.\]
\end{lem}
\begin{proof}
Return to $S_{\pi}$ and set $c=\frac{1}{2}-\frac{1}{2}n^{-\frac{1}{4}}$.
We will show that the state $|A_{\pi}\rangle$ from\prettyref{eq:API}
(with $\lambda$ specified later) furnishes the claimed lower bound
on the ground state energy when plugged into Theorem \ref{thm:collatz-wielandt}.
In order to do this we establish some properties of the set $S_{\pi}$.
First \begin{eqnarray*}
\left|S_{\pi}\right| & = & \sum_{j=1}^{\left\lfloor cn\right\rfloor }\left(\begin{array}{c}
n\\
j\end{array}\right)\\
 & \leq & \left(\frac{1}{c^{c}\left(1-c\right)^{1-c}}\right)^{n}.\end{eqnarray*}
 This upper bound on $\left|S_{\pi}\right|$ follows directly from
the Chernoff-Hoeffding bound. Now let \begin{equation}
f(x)=-x\log_{2}(x)-\left(1-x\right)\log_{2}(1-x)\label{eq:binentropy}\end{equation}
so \[
\left(\frac{1}{c^{c}\left(1-c\right)^{1-c}}\right)^{n}=2^{nf(c)}.\]
 Now \begin{eqnarray*}
f(c) & = & f(\frac{1}{2}-\frac{1}{2}n^{-\frac{1}{4}})=1-\frac{1}{\ln2}\left[\frac{1}{1\cdot2}\left(\frac{1}{n}\right)^{\frac{1}{2}}+\frac{1}{3\cdot4}\left(\frac{1}{n}\right)+\frac{1}{5\cdot6}\left(\frac{1}{n}\right)^{\frac{3}{2}}+...\right]\\
 & \leq & 1-\frac{1}{2\ln2}\left(\frac{1}{\sqrt{n}}\right),\end{eqnarray*}
 so $|S_{\pi}|\leq2^{n\left(1-\frac{1}{2\ln2}\left(\frac{1}{\sqrt{n}}\right)\right)}.$

Lemma \ref{lem:bitlemma} says that the probability $p$ that a set
$\left\{ z,z\oplus e_{j_{1}},z\oplus e_{j_{2}},...,z\oplus e_{j_{k-1}}\right\} \subseteq S_{\pi}$
exists satisfies \begin{eqnarray*}
p & \leq n^{k} & 2^{n\left[1-\frac{k}{2\ln2}\left(\frac{1}{\sqrt{n}}\right)\right]}\\
 & = & 2^{\left[k\frac{\ln n}{\ln2}+n-\frac{k\sqrt{n}}{2\ln2}\right]}.\end{eqnarray*}
We now choose \[
k=\left\lceil \,2\sqrt{n}\,\right\rceil \]
 so that this probability goes to $0$ as $n\rightarrow\infty$. With
high probability there is no set consisting of more than $k-2$ one
bit flip neighbors of $z$ in $S_{\pi}$ when $z$ is in $S_{\pi}.$
So

\[
\text{ Pr}\left[\sum_{i=1}^{n}\frac{\langle z\oplus e_{i}|A_{\pi}\rangle}{\langle z|A_{\pi}\rangle}\leq\frac{\left[\lambda n\left(k-2\right)+1\cdot\left(n-k+2\right)\right]}{\lambda n}\text{ for all }z\in S_{\pi}\right]\rightarrow1\text{ as }n\rightarrow\infty.\]
 Similarly the probability that a set of the form $\left\{ z\oplus e_{i_{1}},z\oplus e_{i_{2}},...,z\oplus e_{i_{k}}\right\} \subseteq S_{\pi}$
exists goes to zero as $n\rightarrow\infty.$ That is, with high probability
there is no set consisting of more than $k-1$ one bit flip neighbors
of $z$ in $S_{\pi}$ for $z\notin S_{\pi}$. So \[
\text{ Pr}\left[\sum_{i=1}^{n}\frac{\langle z\oplus e_{i}|A_{\pi}\rangle}{\langle z|A_{\pi}\rangle}\leq\lambda n\left(k-1\right)+1\cdot\left(n-k+1\right)\text{ for all }z\notin S_{\pi}\right]\rightarrow1\text{ as }n\rightarrow\infty.\]
 Let us compute the lower bound on the ground state energy $E_{g}(s^{\star})$
that is obtained by using $|A_{\pi}\rangle$ in Theorem \ref{thm:collatz-wielandt}
assuming both of these events occur. We get\begin{eqnarray*}
\frac{\langle z|H_{\pi}(s^{\star})|A_{\pi}\rangle}{\langle z|A_{\pi}\rangle} & = & s^{\star}E(z)+\left(\frac{1-s^{\star}}{2}\right)n-\left(\frac{1-s^{\star}}{2}\right)\sum_{i=1}^{n}\frac{\langle z\oplus e_{i}|A_{\pi}\rangle}{\langle z|A_{\pi}\rangle}\\
 & \geq & {\displaystyle \begin{cases}
{\displaystyle \left(\frac{1-s^{\star}}{2}\right)n-\left(\frac{1-s^{\star}}{2}\right)\left[\frac{\lambda n\left(k-2\right)+1\cdot\left(n-k+2\right)}{\lambda n}\right]} & ,\text{ for }z\in S_{\pi}\\
{\displaystyle s^{\star}\cdot nc+\left(\frac{1-s^{\star}}{2}\right)n-\left(\frac{1-s^{\star}}{2}\right)\left[\lambda n\left(k-1\right)+1\cdot\left(n-k+1\right)\right]} & ,\text{ for }z\notin S_{\pi}\;.\end{cases}}\end{eqnarray*}

Now choose \[
\lambda=\frac{1}{k-1}\left(\frac{2s^{\star}c}{1-s^{\star}}-1\right)\]
 so 

\begin{eqnarray*}
\frac{\langle z|H_{\pi}(s^{\star})|A_{\pi}\rangle}{\langle z|A_{\pi}\rangle} & \geq & \begin{cases}
{\displaystyle \left(\frac{1-s^{\star}}{2}\right)n-\left(\frac{1-s^{\star}}{2}\right)\left[\frac{\lambda n\left(k-2\right)+1\cdot\left(n-k+2\right)}{\lambda n}\right]} & ,\text{ for }z\in S_{\pi}\\
{\displaystyle \left(\frac{1-s^{\star}}{2}\right)n+\left(\frac{1-s^{\star}}{2}\right)\left(k-1\right)} & ,\text{ for }z\notin S_{\pi}\end{cases}\\
 & \geq & \left(\frac{1-s^{\star}}{2}\right)\left[n-\left(k-2\right)-\frac{1}{\lambda}\right]\text{ for all }z.\end{eqnarray*}
 Recall that we picked $c=\frac{1}{2}-\frac{1}{2}n^{-\frac{1}{4}}$
and $s^{\star}=\frac{1}{2}+n^{-\frac{1}{4}}.$ This gives \[
\lambda=\frac{1}{k-1}\left(3\left(\frac{1}{n^{\frac{1}{4}}}\right)+4\left(\frac{1}{n^{\frac{1}{4}}}\right)^{2}+8\left(\frac{1}{n^{\frac{1}{4}}}\right)^{3}+...\right)\]
 which means, since $k=\left\lceil \,2\sqrt{n}\,\right\rceil ,$ that
$\lambda\geq\frac{3}{2}n^{-\frac{3}{4}}$ and \[
\frac{\langle z|H_{\pi}(s^{\star})|A_{\pi}\rangle}{\langle z|A_{\pi}\rangle}\geq\left(\frac{1-s^{\star}}{2}\right)\left[n-n^{\frac{3}{4}}\right]\text{ for all }z\]
 with probability approaching $1$ as $n\rightarrow\infty.$ Applying
Theorem \ref{thm:collatz-wielandt} completes the proof of Lemma \ref{lem:Lemma}.
\end{proof}
We are now ready to prove the lower bound on $E_{\pi}(s)$ which is
claimed in Theorem \ref{thm:mainthm}. Since $H_{\pi}(s)$ is linear
in $s$ we can write \[
H_{\pi}(s)=\left(\frac{s-s_{1}}{s_{2}-s_{1}}\right)H_{\pi}(s_{2})+\left(\frac{s_{2}-s}{s_{2}-s_{1}}\right)H_{\pi}(s_{1})\]
 and taking the expectation of both sides in the ground state of $H_{\pi}(s)$
gives, for $s_{1}<s<s_{2},$ \begin{equation}
E_{\pi}(s)\geq\left(\frac{s-s_{1}}{s_{2}-s_{1}}\right)E_{\pi}(s_{2})+\left(\frac{s_{2}-s}{s_{2}-s_{1}}\right)E_{\pi}(s_{1})\label{eq:concavity}\end{equation}
 by the variational principle. Since $E_{\pi}(0)=0$ this gives \[
E_{\pi}(s)\geq\frac{s}{s^{\star}}E_{\pi}(s^{\star}).\]
Using Lemma \eqref{lem:Lemma} we get that with probability $\rightarrow1$
as $n\rightarrow\infty$ \[
E_{\pi}(s)\geq\frac{s}{s^{\star}}\left(\frac{1-s^{\star}}{2}\right)\left[n-n^{\frac{3}{4}}\right]\:\text{ for }0\leq s\leq s^{\star}\]
 and recalling that $s^{\star}=\frac{1}{2}+n^{-\frac{1}{4}}$ gives
\begin{eqnarray*}
E_{\pi}\left({\textstyle \frac{1}{2}}\right) & \geq & \frac{n}{4}\left(\frac{\frac{1}{2}-n^{-\frac{1}{4}}}{\frac{1}{2}+n^{-\frac{1}{4}}}\right)\left(1-n^{-\frac{1}{4}}\right)\\
 & \geq & \frac{n}{4}\left(1-5n^{-\frac{1}{4}}\right).\end{eqnarray*}
Using \eqref{eq:concavity} again twice (with $s_{1}=0$, $s_{2}=\frac{1}{2}$
and with $s_{1}=\frac{1}{2}$ , $s_{2}=1$) we get \[
\frac{E_{\pi}(s)}{n}\geq e(s)\left[1-5n^{-\frac{1}{4}}\right]\]
 for $0\leq s\leq1$ with probability $\rightarrow1$ as $n\rightarrow\infty$,
which completes the proof.

\subsection{Proof of Theorem \ref{thm:gapthm}}

We will use the adiabatic theorem (in the form given in Theorem 3
of \cite{2007JMP....48j2111J} with $m=1$) as well as a theorem from
\cite{2005quant.ph.12159F}. For completeness we reproduce the statements
of these theorems here: 

\begin{adiabaticthm}\label{thm:-adthm}Let $H(s)$ be a finite-dimensional
twice differentiable Hamiltonian on $0\leq s\leq1$ with a nondegenerate
ground state $|\phi(s)\rangle$ separated by an energy gap $\gamma(s).$
Let $|\psi(t)\rangle$ be the state obtained by Schrödinger time evolution
with Hamiltonian $H\left(\frac{t}{T}\right)$ starting with state
$|\phi(0)\rangle$ at $t=0.$ Then \[
\sqrt{1-\left|\langle\psi(T)|\phi(1)\rangle\right|^{2}}\leq\frac{1}{T}\left[\frac{1}{\gamma(0)^{2}}\left\Vert \frac{dH}{ds}\right\Vert _{s=0}+\frac{1}{\gamma(1)^{2}}\left\Vert \frac{dH}{ds}\right\Vert _{s=1}+\int_{0}^{1}ds\left(\frac{7}{\gamma^{3}}\left\Vert \frac{dH}{ds}\right\Vert ^{2}+\frac{1}{\gamma^{2}}\left\Vert \frac{d^{2}H}{ds^{2}}\right\Vert \right)\right].\]

\end{adiabaticthm}

The next theorem considers the complexity of finding the minimum of
any scrambled cost function by continuous time Hamiltonian evolution.
The quantum adiabatic algorithm is only a special case of continuous
time Hamiltonian evolution. The theorem says that for a totally unstructured,
i.e scrambled, cost function no quantum algorithm can achieve more
than Grover speed up.

\begin{scramblethm}\label{thm:scramthm}Let $h(z)$ be a cost function,
with $h(0)=0$ and $h(1),h(2),...,h(N-1)$ all positive. Let $\pi$
be a permutation on $N$ elements, and $H_{D}(t)$ be an arbitrary
$\pi$-independent Hamiltonian. Consider the Hamiltonian \[
\tilde{H}_{\pi}(t)=H_{D}(t)+c(t)\left(\sum_{z=0}^{N-1}h(\pi^{-1}(z))|z\rangle\langle z|\right),\]
 where $|c(t)|\leq1$ for all $t$. Let $|\psi_{\pi}(T)\rangle$ be
the state obtained by Schrodinger evolution governed by $\tilde{H}_{\pi}(t)$
for time $T$, with a $\pi$-independent starting state. Suppose that
the success probability $|\langle\psi_{\pi}(T)|\pi(0)\rangle|^{2}\geq\frac{1}{2}$,
for a set of $\epsilon N!$ permutations. Then \begin{eqnarray*}
T & \geq & \frac{\epsilon^{2}\sqrt{N}}{64h^{\star}}\;\text{ for }N\geq\frac{256}{\epsilon}.\end{eqnarray*}
 where \[
h^{\star}=\left(\frac{\sum_{z}h(z)^{2}}{N-1}\right)^{\frac{1}{2}}.\]

\end{scramblethm}

We now return to the quantum adiabatic Hamiltonian $H_{\pi}(s)$ given
by\prettyref{eq:hpi} and apply the Adiabatic Theorem. We have \[
\gamma_{\pi}(0)=\gamma_{\pi}(1)=1\]
 and \[
\left\Vert \frac{dH}{ds}\right\Vert \leq2n.\]
 Here $N=2^{n}$. Furthermore the ground state at $s=1$ is $|\phi(1)\rangle=|\pi(0)\rangle$.
The state $|\psi_{\pi}\left(T\right)\rangle$ is obtained by evolving
with the Hamiltonian $H_{\pi}\left(\frac{t}{T}\right)$ starting from
the state $|x=0\rangle$. Plugging into the Adiabatic Theorem we get\begin{eqnarray}
\sqrt{\left(1-\left|\langle\psi_{\pi}(T)|\pi(0)\rangle\right|^{2}\right)} & \leq & \frac{1}{T}\left[4n+\int_{0}^{1}\frac{28n^{2}}{\gamma_{\pi}^{3}}ds\right]\nonumber \\
 & \leq & \frac{1}{T}\left[4n+\frac{28n^{2}}{g_{\pi}^{3}}\right]\nonumber \\
 & \leq & \frac{32n^{2}}{Tg_{\pi}^{3}}\label{eq:ineq}\end{eqnarray}
where in the last line we used the fact that $g_{\pi}\leq\gamma_{\pi}(0)=1.$ 

Fix $0<\epsilon<1$. Let $R$ be the set of permutations $\pi$ for
which \begin{equation}
g_{\pi}>\left(\frac{32n^{3}}{\left[\frac{\epsilon{}^{2}}{128}\sqrt{N}\right]}\right)^{\frac{1}{3}}.\label{eq:contradict}\end{equation}
 Then plugging into\prettyref{eq:ineq} we get \[
\sqrt{\left(1-\left|\langle\psi_{\pi}(T)|\pi(0)\rangle\right|^{2}\right)}\leq\frac{\left(\frac{\epsilon{}^{2}}{128n}\sqrt{N}\right)}{T}\text{ for all }\pi\in R.\]
 Now choose \begin{equation}
T=\frac{\sqrt{2}\epsilon{}^{2}}{128n}\sqrt{N}\label{eq:tstar}\end{equation}
 which guarantees that $\left|\langle\psi_{\pi}(T)|\pi(0)\rangle\right|^{2}\geq\frac{1}{2}$
for all $\pi\in R$. Assume (to get a contradiction) that the size
of $R$ is at least $\epsilon N!$. Now apply the Scrambled Theorem
to obtain (for $N\geq\frac{256}{\epsilon}$) \begin{eqnarray*}
T & \geq & \frac{\epsilon^{2}}{64n}\sqrt{N}\end{eqnarray*}
 where we have used the fact that in our case $h^{\star}\leq n.$
But\prettyref{eq:tstar} contradicts the last inequality. Therefore
$R$ cannot contain $\epsilon N!$ permutations when $N$ is sufficiently
large. In other words, for $N\geq\frac{256}{\epsilon}$ \[
\text{Pr}\left[g_{\pi}\leq\left(\frac{32n^{3}}{\left[\frac{\epsilon{}^{2}}{128}\sqrt{N}\right]}\right)^{\frac{1}{3}}\right]\geq1-\epsilon.\]
 For example choosing $\epsilon=\frac{64}{n\sqrt{n}}$ we get Theorem
\ref{thm:gapthm} in the stated form.

\subsection{Proof of Theorem \ref{thm:nopert_cross}}

For any Hamiltonian $H$, the first excited state energy $E_{1}$
is equal to \[
E_{1}=\max_{|\phi\rangle}\left[\min_{|\psi\rangle\text{ s.t }\langle\phi|\psi\rangle=0}\;\langle\psi|H|\psi\rangle\right]\]
where $\langle\psi|\psi\rangle=1$. Let us now apply this fact to
bound the first excited state energy $E_{1,\pi}(s)$ of $H_{\pi}(s)$.
We can get a lower bound on $E_{1,\pi}(s)$ by fixing $|\phi\rangle$
to be a particular state. Choosing $|\phi\rangle=|\pi(0)\rangle$
gives \[
E_{1,\pi}(s)\geq\min_{|\psi\rangle\text{ s.t }\langle\pi(0)|\psi\rangle=0}\;\langle\psi|H_{\pi}(s)|\psi\rangle.\]
 Now the quantity on the RHS is the ground state energy of the Hamiltonian
obtained from $H_{\pi}(s)$ by removing one row and one column corresponding
to the state $|\pi(0)\rangle.$ This reduced Hamiltonian is stoquastic
and applying Theorem \ref{thm:collatz-wielandt} to it gives a lower
bound \begin{equation}
E_{1,\pi}(s)\geq\min_{z\text{\ensuremath{\neq\pi}(0)}}\frac{\langle z|H_{\pi}(s)|\chi\rangle}{\langle z|\chi\rangle}\label{eq:firstexcitedbnd}\end{equation}
 for any state $|\chi\rangle$ such that $\langle y|\chi\rangle>0$
for all bit strings $y\neq\pi(0)$ and $\langle\pi(0)|\chi\rangle=0$.

We will use the bound\prettyref{eq:firstexcitedbnd} with a state
$|\chi_{\pi}\rangle$ defined by \[
\langle z|\chi_{\pi}\rangle=\begin{cases}
1 & ,\: z\notin\tilde{S}_{\pi}\\
\mu n & ,\: z\in\tilde{S}_{\pi}\\
0 & ,\: z=\pi(0).\end{cases}\]
 where $\mu$ depending on $s$ will be chosen later and \[
\tilde{S}_{\pi}=\left\{ z:\: z\neq\pi(0)\text{ and }W(\pi^{-1}(z))\leq nc\right\} \]
 for some $c$ which we will also choose later. $\tilde{S}_{\pi}$
is the image of the set $\left\{ z:\: z\neq0\text{ and }W(z)\leq nc\right\} $
under the permutation $\pi$. For $n$ large enough we will have $\mu n\geq1$.

To evaluate the RHS of\prettyref{eq:firstexcitedbnd} we need \begin{eqnarray*}
\frac{\langle z|H_{\pi}(s)|\chi_{\pi}\rangle}{\langle z|\chi_{\pi}\rangle} & = & sW\left(\pi^{-1}(z)\right)+\left(\frac{1-s}{2}\right)n-\left(\frac{1-s}{2}\right)\sum_{i=1}^{n}\frac{\langle z\oplus e_{i}|\chi_{\pi}\rangle}{\langle z|\chi_{\pi}\rangle}\\
 & \geq & \begin{cases}
{\displaystyle s+\left(\frac{1-s}{2}\right)n-\left(\frac{1-s}{2}\right)\sum_{i=1}^{n}\frac{\langle z\oplus e_{i}|\chi_{\pi}\rangle}{\langle z|\chi_{\pi}\rangle}} & ,\text{ for }z\in\tilde{S}_{\pi}\\
{\displaystyle s\cdot nc+\left(\frac{1-s}{2}\right)n-\left(\frac{1-s}{2}\right)\sum_{i=1}^{n}\frac{\langle z\oplus e_{i}|\chi_{\pi}\rangle}{\langle z|\chi_{\pi}\rangle}} & ,\text{ for }z\notin\tilde{S}_{\pi}\text{ and }z\neq\pi(0).\end{cases}\end{eqnarray*}
The set $\tilde{S}_{\pi}$ has cardinality \[
\left|\tilde{S}_{\pi}\right|\leq2^{f(c)n}\]
where $f(x)$ is given by\prettyref{eq:binentropy}. Applying Lemma
\ref{lem:bitlemma} with $k=2$ we get that with probability at least
$1-n^{2}\cdot2^{n\left[1-2\left(1-f(c)\right)\right]}$ there is no
pair of one bit flip neighbors \[
\left\{ z,z\oplus e_{j}\right\} \subseteq\tilde{S}_{\pi}\]
 and similarly with probability at least $1-n^{2}\cdot2^{n\left[1-2\left(1-f(c)\right)\right]}$
there is no set \[
\left\{ z\oplus e_{j_{1}},z\oplus e_{j_{2}}\right\} \subseteq\tilde{S}_{\pi}.\]
Choose $c<\frac{1}{2}$ to make $f(c)=0.49$. Then with probability
at least $1-2n^{2}2^{-0.02n}$ we have \begin{eqnarray}
\frac{\langle z|H_{\pi}(s)|\chi_{\pi}\rangle}{\langle z|\chi_{\pi}\rangle} & \geq & \begin{cases}
s+\left(\frac{1-s}{2}\right)n-\left(\frac{1-s}{2}\right)\frac{1}{\mu} & ,\text{ for }z\in\tilde{S}_{\pi}\\
s\cdot nc+\left(\frac{1-s}{2}\right)n-\left(\frac{1-s}{2}\right)\left[n\mu+n-1\right] & ,\text{ for }z\notin\tilde{S}_{\pi}\text{ , }z\neq\pi(0)\end{cases}\nonumber \\
 & = & \begin{cases}
\left(\frac{1-s}{2}\right)n-\left(\frac{1-s}{2}\right)\left[\frac{1}{\mu}-\frac{2s}{1-s}\right] & ,\text{ for }z\in\tilde{S}_{\pi}\\
\left(\frac{1-s}{2}\right)n-\left(\frac{1-s}{2}\right)\left[n\left(\mu+1\right)-1-\frac{2snc}{1-s}\right] & ,\text{ for }z\notin\tilde{S}_{\pi}\text{ , }z\neq\pi(0).\end{cases}\label{eq:b3}\end{eqnarray}
 For $s>\frac{1}{1+2c}\approx0.82$ choose $\mu$ to be \[
\mu=\frac{2sc}{1-s}-1.\]
 Then \begin{eqnarray*}
\left(\frac{1-s}{2}\right)\left[n\left(\mu+1\right)-1-\frac{2snc}{1-s}\right] & = & -\left(\frac{1-s}{2}\right)\end{eqnarray*}
 and \begin{eqnarray}
\left(\frac{1-s}{2}\right)\left[\frac{1}{\mu}-\frac{2s}{1-s}\right] & = & \frac{1-s^{2}\left(1+4c\right)}{2s\left(1+2c\right)-2}.\label{eq:b2}\end{eqnarray}
 Plugging this into\prettyref{eq:b3} gives that with probability
at least $1-n^{2}2^{-0.02n}$ , for $0.9\leq s\leq1$, \[
E_{1,\pi}(s)\geq\min_{z\neq\pi(0)}\frac{\langle z|H_{\pi}(s)|\chi_{\pi}\rangle}{\langle z|\chi_{\pi}\rangle}\geq\left(\frac{1-s}{2}\right)n+\frac{s^{2}\left(1+4c\right)-1}{2s\left(1+2c\right)-2}.\]
 Then since the ground state energy of $H_{\pi}(s)$ is less than
or equal to $\left(\frac{1-s}{2}\right)n$, the gap is bounded below
by the second term on the RHS, with probability at least $1-n^{2}2^{-0.02n}.$
For $0.9\leq s\leq1$ this term is larger than $0.8$. This completes
the proof.

\section*{Acknowledgements}

We thank Elihu Abraham, Boris Altshuler, Chris Laumann, and Vadim
Smelyanskiy for interesting discussions. This work was supported in
part by funds provided by the W. M. Keck Foundation Center for Extreme
Quantum Information Theory, the U.S Army Research Laboratory's Army
Research Office through grant number W911NF-09-1-0438, the National
Science Foundation through grant number CCF-0829421, and the Natural
Sciences and Engineering Research Council of Canada.

\bibliographystyle{plain}
\nocite{*}
\bibliography{straightlines}

\end{document}